\documentclass[runningheads,a4paper]{llncs}

\usepackage[T1]{fontenc}
\usepackage{graphicx}

\long\def\gobbleone#1{}
\protected\def\atan{\futurelet\tmptoken\doatan}
\protected\def\doatan{\operatorname{atan\ifx\tmptoken22\fi}%
  \ifx\tmptoken2\expandafter\gobbleone\fi}
\protected\def\arctan{\futurelet\tmptoken\doarctan}
\protected\def\doarctan{\operatorname{arc\,tan\ifx\tmptoken22\fi}%
  \ifx\tmptoken2\expandafter\gobbleone\fi}

\usepackage{hyperref}
\usepackage{cite}
\usepackage{amsmath}
\usepackage{amssymb}
\usepackage{mathabx}

\urldef{\mailsa}\path|grzegorz.borowik@tricholab.com|
\urldef{\mailsb}\path|michal.balicki.job@gmail.com|
\urldef{\mailsc}\path|michal.kasprzak@tricholab.com|
\urldef{\mailsd}\path|piotr.cukier@tricholab.com|

\newcommand{\keywords}[1]{\par\addvspace\baselineskip
\noindent\keywordname\enspace\ignorespaces#1}

\begin{document}

\mainmatter 

\title{Improved Mesh Processing using Distorted Pole Spherical Coordinates}

\titlerunning{Improved Mesh Processing using Distorted Pole Spherical Coordinates}

\author{Grzegorz Borowik$^{1,2}$
        \and Micha\l{} Balicki$^{1}$
	\and Micha\l{} Kasprzak$^{1}$
        \and Piotr Cukier$^{1}$
}

\authorrunning{G. Borowik, M. Balicki, M. Kasprzak, and P. Cukier
}

\institute{
$^{1}$
TrichoLAB LLC, Orzycka 27, 02-695 Warsaw, Poland\\
$^{2}$
University of Technology and Economics, Warsaw, Poland
\mailsa, \mailsb, \mailsc, \mailsd\\
ORCID: 0000-0003-4148-4817\\
}

\maketitle

\begin{abstract}
The Cartesian coordinate system is the most commonly used system in computer visualization.
This is due to its ease of use and processing speed. However, it is not always suitable for a given problem. Angular measures often allow us to operate more efficiently on a three-dimensional model.
When dealing with issues related to the processing of three-dimensional objects described using polygon meshes, it often happens that these standard systems do not satisfy specific properties that are crucial to us.
The topic of the paper is to discuss a specific transformation to spherical coordinates with distorted poles, which allows us to eliminate singular points from the determined subset of the mesh and bypass inconvenient seam lines in its two-dimensional projection, which can hinder further calculations.

\end{abstract}

\section{Introduction}
Polygon meshes are commonly used to represent 3D objects in computer graphics applications. The mesh is composed of a set of vertices, edges, and faces that define the shape of the object. The vertices are represented by their spatial coordinates in 3D space, which are typically specified in a Cartesian coordinate system. However, Cartesian coordinates may not always be the best choice for certain applications. This paper discusses an alternative coordinate system that is better suited for some specific problems.

While Cartesian coordinates are simple to use and process quickly, they have some limitations and are not always suitable for describing certain types of shapes, such as spherical or cylindrical shapes. A limitation is, for example, that they may not be able to preserve some of the important geometric properties of the original shape.

On the other hand, angular measurements can provide a more efficient way to manipulate 3D models. For example, when dealing with spherical shapes, it is often more convenient to use spherical coordinates instead of Cartesian coordinates. Moreover, using spherical coordinates avoids the need for complex trigonometric calculations, and can simplify various transformations. 

When working with polygon meshes, it is not uncommon to encounter singular points. Singular points are points where the mesh is distorted, and the geometry of the surrounding points cannot be well defined. Singular points can cause problems in many types of mesh processing operations, such as subdivision or smoothing.

This paper discusses a specific transformation to spherical coordinates that can eliminate singular points from a subset of the mesh. The transformation involves introducing a distortion to the poles of the spherical coordinate system, which effectively removes the singular points. The transformed mesh can then be represented using a regular spherical coordinate system, which is more suitable for further processing.

Another problem with polygon meshes is that their two-dimensional projections may have inconvenient seam lines. Seam lines are lines that mark the boundary between two parts of the mesh that are mapped to the same area on the two-dimensional surface. Seam lines can cause problems in texture mapping or other types of surface analysis. This paper presents a method for avoiding seam lines by carefully mapping the surface of the mesh to a two-dimensional theta-phi space.

In summary, this paper presents a novel application of spherical transformation with distorted poles. The method enables easy control over the singular points and avoids seam lines, by granting the ability to move them out of the domain of interest. The proposed technique can be useful for various applications in computer graphics involving star-shaped polygon meshes. 

\section{Related work}
In recent years, several techniques have been developed to overcome challenges in representing and manipulating three-dimensional data in computer graphics. Taylor et al. proposed a method that uses standard computer graphics techniques to avoid issues with three-dimensional data in arbitrary coordinate systems~\cite{Taylor2017Visualizing}.

One common approach for representing three-dimensional objects is to use mesh data structures. Hoppe et al. introduced a wedge-based mesh data structure that captures attribute discontinuities efficiently and allows for simultaneous optimization of multiple attribute vectors per vertex~\cite{HoppeNew}. Meanwhile, Johnson et al. proposed a sequence of local edge operations that promote uniform edge lengths while preserving mesh shape~\cite{Johnson1998Control}.

Cartesian coordinate systems have also been used in computer graphics to simplify trigonometry. For instance, Deal recommended a particular Cartesian coordinate system of adding lenses~\cite{DEAL1993Recommended}.

Barycentric coordinates have been widely used in planar parameterization, but Gotsman et al. generalized this method to solve the spherical parameterization problem~\cite{Gotsman2003Fundamentals}. Meanwhile, Kahler discusses the use of triangle meshes to represent geometric surfaces, which are simple but require a large number of triangles to represent a smooth surface. To address this, mesh simplification algorithms have been developed to reduce the number of triangles while approximating the initial mesh~\cite{KahlerEfficient}.

Mesh representation has been another area of active research in computer graphics. Birthelmer et al. proposed a method of storing vertices and topology of meshes independently of each other~\cite{Birthelmer2003Efficient}. Lipman et al. introduced a linear least squares system that can be solved quickly enough for interactive response time~\cite{LipmanDifferential}. They also proposed a rigid motion invariant mesh representation based on discrete forms defined on the mesh~\cite{Lipman2005Linear}, which is unique to a rigid transformation of the mesh.

A variety of other methods have been proposed in computer graphics to address challenges in representing and manipulating three-dimensional data, ranging from mesh data structures to coordinate systems and parameterization techniques.

Especially, \cite{CoordinateTransformationonaSphereUsingConformalMapping} describes an analytical reversible coordinate transformation on a sphere that can overcome the problem of using traditional spherical coordinates that have a singularity at the pole, which is important for setting up global ocean circulation models that include the Arctic Ocean. The new transformation preserves the orthogonality and the shape of infinitesimal figures and can map the North and South Poles to two arbitrary locations on the Earth using conformal mapping in the extended complex plane. The resulting coordinate system has enhanced resolution along the geodesic curve between the new poles, and the transformation has been used to increase the resolution in a specific area of interest in the model domain. This concept has been used to build our improved mesh processing method.

\section{Fundamentals}
In the context of studying spherical mappings and projections, it is useful to define the notion of a star-shaped polygon mesh. This concept plays a key role in the analysis of certain geometric structures and their properties. We provide the following definition:

\begin{definition} 
{\upshape Let $f : \mathbb{R}^3 \longrightarrow [0, \pi] \times [-\pi, \pi] \times \mathbb{R}$ be a spherical mapping. Let~$g$ be a projection on the theta-phi space.
Then mesh $\mathcal{M} \subseteq \mathbb{R}^3$ is called {\itshape star-shaped polygon mesh} if and only if 
$\forall_{x, y \in \mathcal{M}} \; g(f(x)) = g(f(y)) \iff x = y$.}
\end{definition}

The methodology considers exclusively star-shaped polygon meshes as they imply the unique representation of the points in the theta-phi space, allowing for further processing in the two-dimensional space. Our mapping is performed using the standard spherical mapping. Namely for $f(x, y, z) = (\theta, \phi, r)$ we have:

\begin{equation}\label{cartesian_spherical}
    \theta  =  \atan2(\sqrt{x^2 + y^2}, z), \ \
    \phi  =  \atan2(y, x), \ \
    r  =  \sqrt{x^2 + y^2 + z^2}
\end{equation}
\ \\

The inverse transformation: $f^{-1}(\theta, \phi, r) = (x, y, z)$, given by (\ref{spherical_cartesian}), is also a~well known mapping. Both transformations enable us to freely switch between the two coordinate systems, which will be useful in later considerations.
\begin{equation}\label{spherical_cartesian}
\begin{array}{rcl}
    x  = r \cos(\theta) \sin(\phi),\ \    y  =  r \sin(\theta) \sin(\phi),\ \   z  =  r \cos(\phi)
\end{array}
\end{equation}

The formula (\ref{main_transformation}) is a combination of the stereographic polar transformation composed with the conformal mapping which is responsible for warping the poles of the sphere, bringing them closer to each other. Due to the nature of transformations circles are mapped to circles, preserving the longitudes and latitudes of the sphere. For $\sigma = 0.3$ we obtain the transformation visualized in Fig.~\ref{fig:stereographicProj} for 2D and in Fig.~\ref{fig:3DWarpedPoles} for 3D.

\begin{equation}\label{main_transformation}
    \hat{x} = x_\sigma * \delta,\ \
    \hat{y} = 1 - \delta,\ \
    \hat{z} = z_\sigma * \delta,
\end{equation}
where
$\delta = \displaystyle\frac{2} { (1 + x_\sigma^2 + z_\sigma^2)}$, and $x_\sigma = \displaystyle\frac{\sigma x }{ 1 - y}$, $z_\sigma = \displaystyle\frac{\sigma z }{ 1 - y}.$

\begin{figure}[h]
    \centering
    \includegraphics[width=0.95\textwidth]{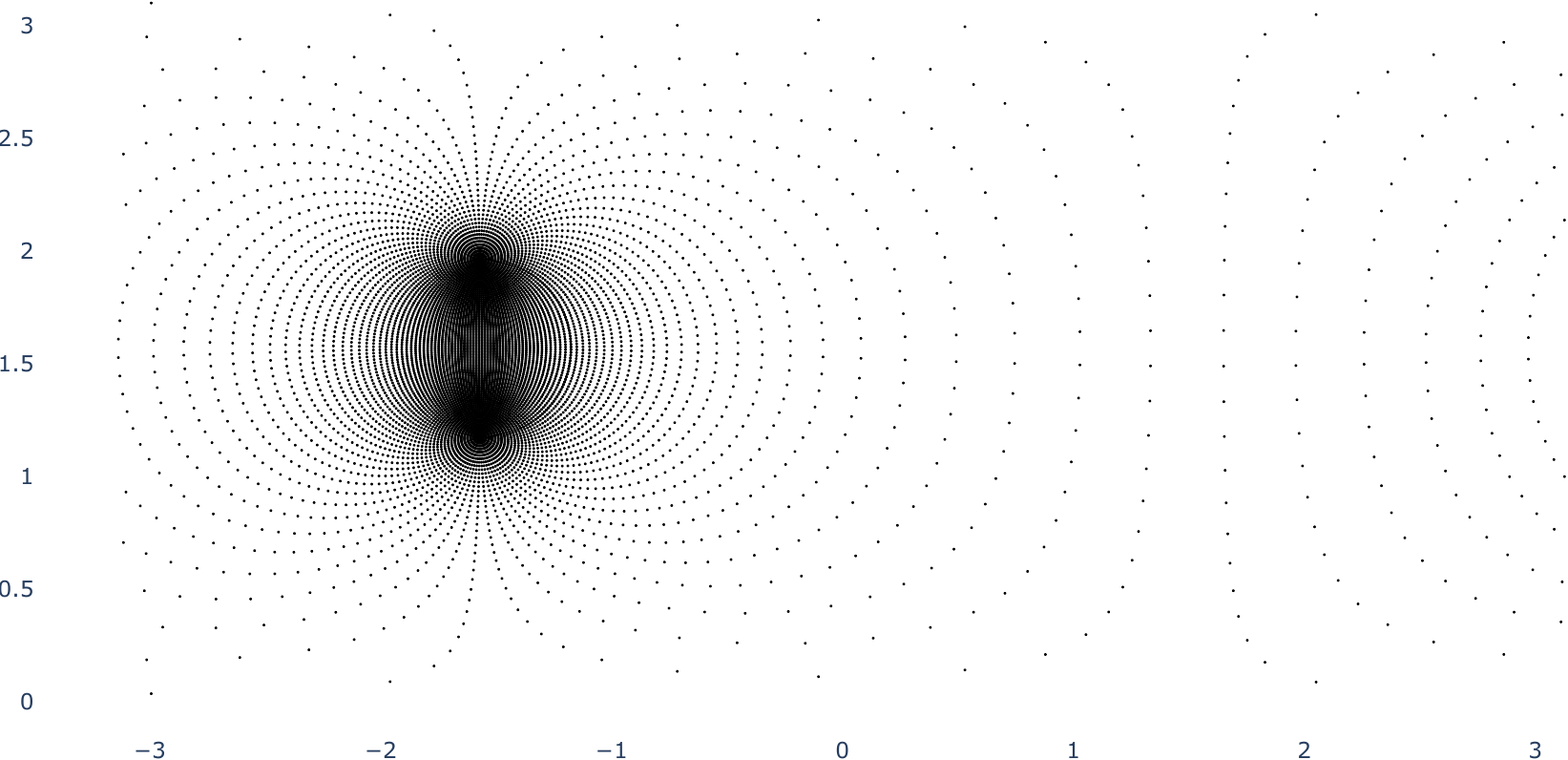}
    \caption{2D visualization of the spherical warp (theta-phi)}
    \label{fig:stereographicProj}
\ \\
\ \\
    \centering
    \includegraphics[width=0.3\textwidth]{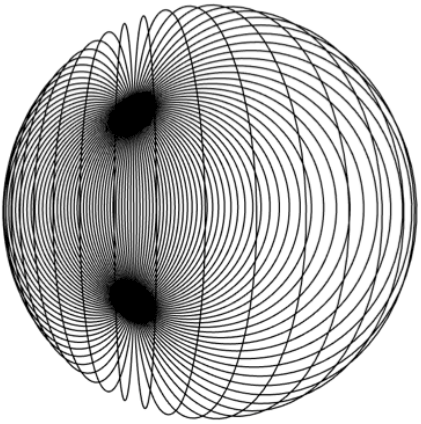}
    \caption{3D plot of the sphere with wrapped poles}
    \label{fig:3DWarpedPoles}
\end{figure}

\section{Methodology}
Mapping the two-dimensional sub-area of the surface of a star-shaped model can be described by the composition of the provided transformations. However, before applying those mappings some prerequisites need to be met.

First, the origin of the system, i.e. point $[0,0,0]$ in the Cartesian coordinates, must be the point relative to which the model is star-shaped. Sometimes it might be necessary to translate the coordinate system to that position.

Second, it is important to rotate the model so that the poles are relatively close to the boundaries of the two-dimensional sub-area. This ensures an accurate representation of the geometric structures of the mesh after transformation. For our purposes, we constrain these rotations to work with the natural symmetry of the 3D model. However, in general, this is not required.

The distortion task uses (\ref{main_transformation}) to warp the poles and then (\ref{cartesian_spherical}) to transform the result to the spherical coordinate system. The distortion parameter should be chosen carefully. Keeping it too small or too large makes the two-dimensional representation less accurate.

Sometimes, it might be useful to go back to the standard spherical space. In such a scenario, we can use (\ref{spherical_cartesian}) and (\ref{main_transformation}) with the inverse of the distortion parameter $\sigma$, and (\ref{cartesian_spherical}) to return to the original theta-phi space. The inverse transformation is the result of the Theorem~\ref{inv-therem}.
\begin{theorem}\label{inv-therem} 
Let $f(x, y, z, \sigma) = (\hat{x}, \hat{y}, \hat{z})$ be a cartesian warp transformation described by (\ref{main_transformation}). Then, 
\begin{equation}
    \forall_{\sigma \in \mathbb{R}} \forall_{(x,y,z) \in K((0,0), 1)} \; f(f(x, y, z, \sigma), \frac{1}{\sigma}) = (x, y, z)
\end{equation}
\end{theorem}

\begin{proof}
From (\ref{main_transformation}) we get:
\begin{eqnarray*}
    \hat{x} &=&  \frac{-2\sigma  x( y - 1)}{\sigma^2 x^2 + \sigma^2  z^2 + y ^2 - 2  y + 1} \\
    \hat{y} &=& \frac{\sigma ^2  x^2 + \sigma ^2  z ^2 - y ^2 + 2  y - 1} {\sigma ^2  x ^2 + \sigma ^2  z^2 + y ^2 - 2  y + 1} \\
    \hat{z} &=& \frac{-2 \sigma z (y - 1)}{\sigma ^2 x^2 + \sigma ^2  z ^2 + y ^2 - 2  y + 1}
\end{eqnarray*}
then,
\begin{eqnarray*}
\displaystyle \dot{x} =  \frac{-\frac{2}{\sigma}  \hat{x}(\hat{y} - 1)}  {\frac{1}{\sigma^2}    \hat{x}^2 + \frac{1}{\sigma^2}    \hat{z}^2 + \hat{y}^2  - 2  \hat{y} + 1} &=& \frac{-2x(y - 1)}{x^2 + y^2 + z^2 - 2y + 1}\\
\displaystyle \dot{y} = \frac{\frac{1}{\sigma^2}  \hat{x}^2 +  \frac{1}{\sigma^2}   \hat{z}^2 -  \hat{y}^2 + 2  \hat{y} - 1}  {\frac{1}{\sigma^2}   \hat{x}^2  + \frac{1}{\sigma^2}   \hat{z}^2  + \hat{y}^2  - 2  \hat{y} + 1} &=& \frac{x^2 - y^2 + z^2 + 2y - 1}{x^2 + y^2 + z^2 - 2y + 1}\\
\displaystyle \dot{z} =  \frac{-\frac{2}{\sigma}\hat{z}(  \hat{y} - 1)}    {\frac{1}{\sigma^2}    \hat{x}^2 + \frac{1}{\sigma^2} \hat{z}^2  +  \hat{y}^2 - 2  \hat{y} + 1} &=& \frac{-2z(y - 1)}{x^2 + y^2 + z^2 - 2y + 1}
\end{eqnarray*}
\qed
\end{proof}

\section{Alternative approach}
The approach given by (\ref{main_transformation}) is not the only possibility for achieving the desired results. We propose an alternative way of addressing the problem of distorting poles, which takes inspiration from observing the different plane projections of the previously described transformation. The main advantage is the ability to calculate $\theta$ and $\phi$ coordinates separately, unlike the described mapping. 

\begin{definition}
{\upshape Let $f_{\sigma}$ : $\mathbb{R}^3 \longrightarrow [\theta, \phi]$ such that:
\begin{equation}
     \theta = f_{1, \sigma}(x, y, z),\ \ \phi = f_{2, \sigma}(x, y, z).
\end{equation}
By $f_{\sigma}$ we denote a function that maps the Cartesian coordinates to the theta-phi space with distorted poles based on the $\sigma$ parameter. $f_{\sigma}$ consists of two coordinate-wise functions, which calculate $\theta$ and $\phi$ respectively in the new system.
} 
\end{definition}

Let's examine the $[p_x, p_y, p_z]$ coordinates of (\ref{main_transformation}) projected onto the $XY$ plane. The latitudes in spherical coordinates with distorted poles get mapped to ellipses for which $\theta$ angle is preserved (Fig.~\ref{fig:xy-projection}).

By taking the anticlockwise angle measurement between the vectors $[p_x -o_x, p_y - o_y]$ and $[1, 0]$ we get the $\theta$ parameter of our new transformation, where $[o_x, o_y]$ are the coordinates of the new pole. Since the distortion happens along the y-axis, $o_x = 0$, calculations are very efficient.

The difference between this measure and the approach given by (\ref{main_transformation}) is that we ignore the effect of $\phi$ on the latitudinal measure, which greatly simplifies the calculation.

\begin{theorem} 
By $f_{1, \sigma}: \mathbb{R}^3 \longrightarrow [0, 2\pi]$ let's denote the function that returns $\theta$ from spherical coordinates with distorted poles.
Let $O = [o_x, o_y, o_z]$ be a Cartesian representation of a distorted north pole, then:
\begin{equation}
\theta = f_{1, \sigma}([p_x, p_y, p_z]) = \operatorname{atan2}(p_y - oy, p_x)
\end{equation}
\end{theorem}

\begin{proof}
$$\displaystyle [1,0] \cdot ([p_x, p_y] - [o_x, o_y]) = p_x - o_x = p_x$$
$$\displaystyle \begin{vmatrix} 1 & p_x - o_x\\0 & p_y - o_y \end{vmatrix} = p_y - o_y$$
$$\displaystyle \theta =  \arccos({\frac{p_x}{||[p_x, p_y] - [o_x, o_y]|| }})\operatorname{sgn}(p_y - o_y)$$
$$\displaystyle \theta = \arccos(\frac{p_x}{\sqrt{p_x^2 + (p_y - o_y)^2}})\operatorname{sgn}(p_y - o_y)$$
$$\displaystyle \theta = \operatorname{atan2}(p_y - oy, p_x)$$
 \qed
\end{proof}

Looking at the projection on the plane $YZ$ (Fig.~\ref{fig:yz-projection}) of the point $[p_x, p_y, p_z]$ from transformation (\ref{main_transformation}), we can see that the latitudes are projected onto chords of the unit circle which, as $\phi$ approaches the pole, converge to the point $(o_y, o_z)$.
By finding the chord in which the $[p_y, p_z]$ lies, we can evaluate the value of~$\phi$. The calculation of $\phi$ relies on the angular measure ratio between the points defining the chord and the corresponding angular measure of the warped poles.

\begin{figure}[ht]
\centering
\begin{minipage}{.5\textwidth}
  \centering
  \includegraphics[width=0.7\textwidth]{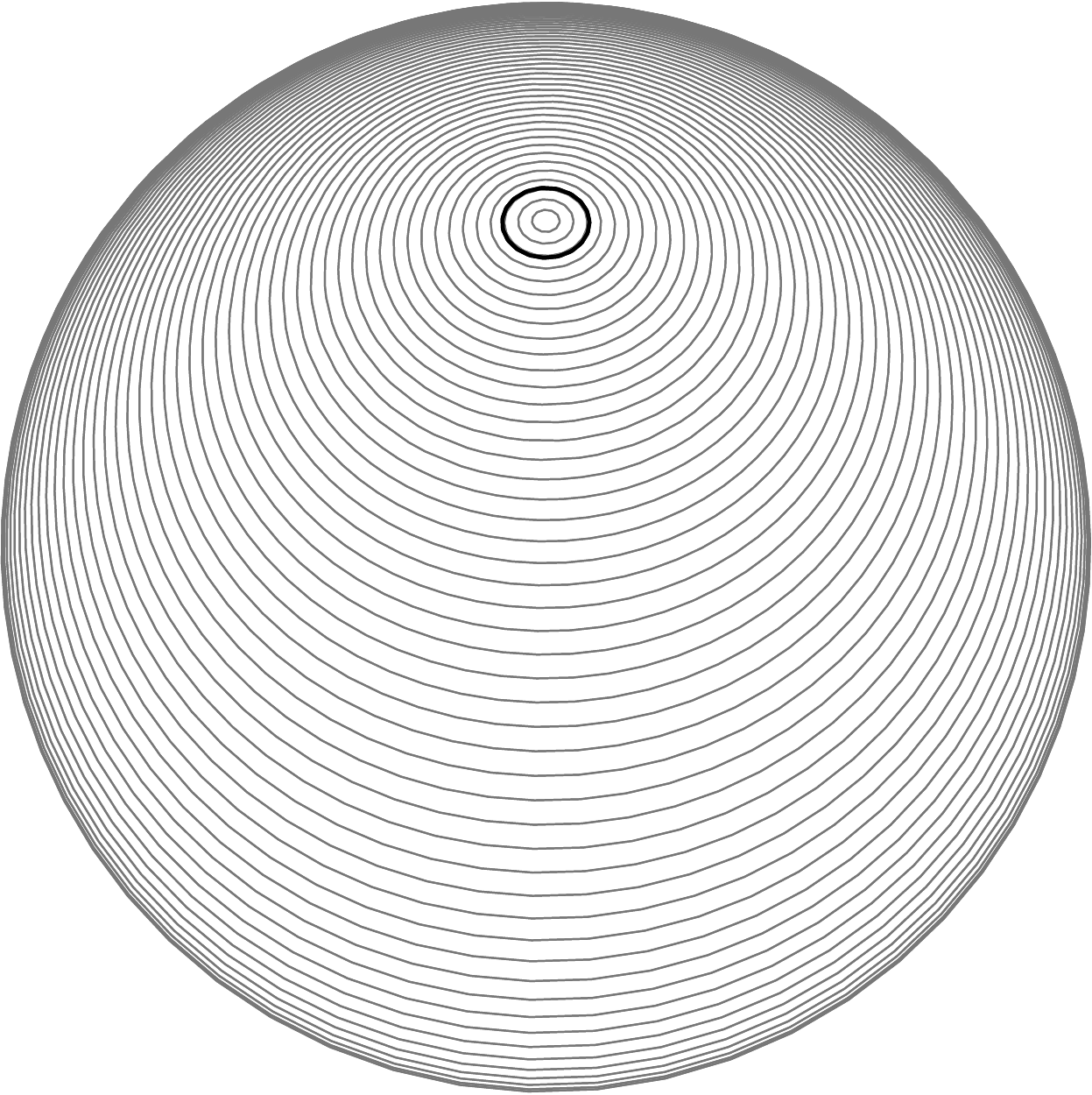}
    \caption{$XY$ projection}
    \label{fig:xy-projection}
\end{minipage}%
\begin{minipage}{.5\textwidth}
  \centering
  \includegraphics[width=0.7\textwidth]{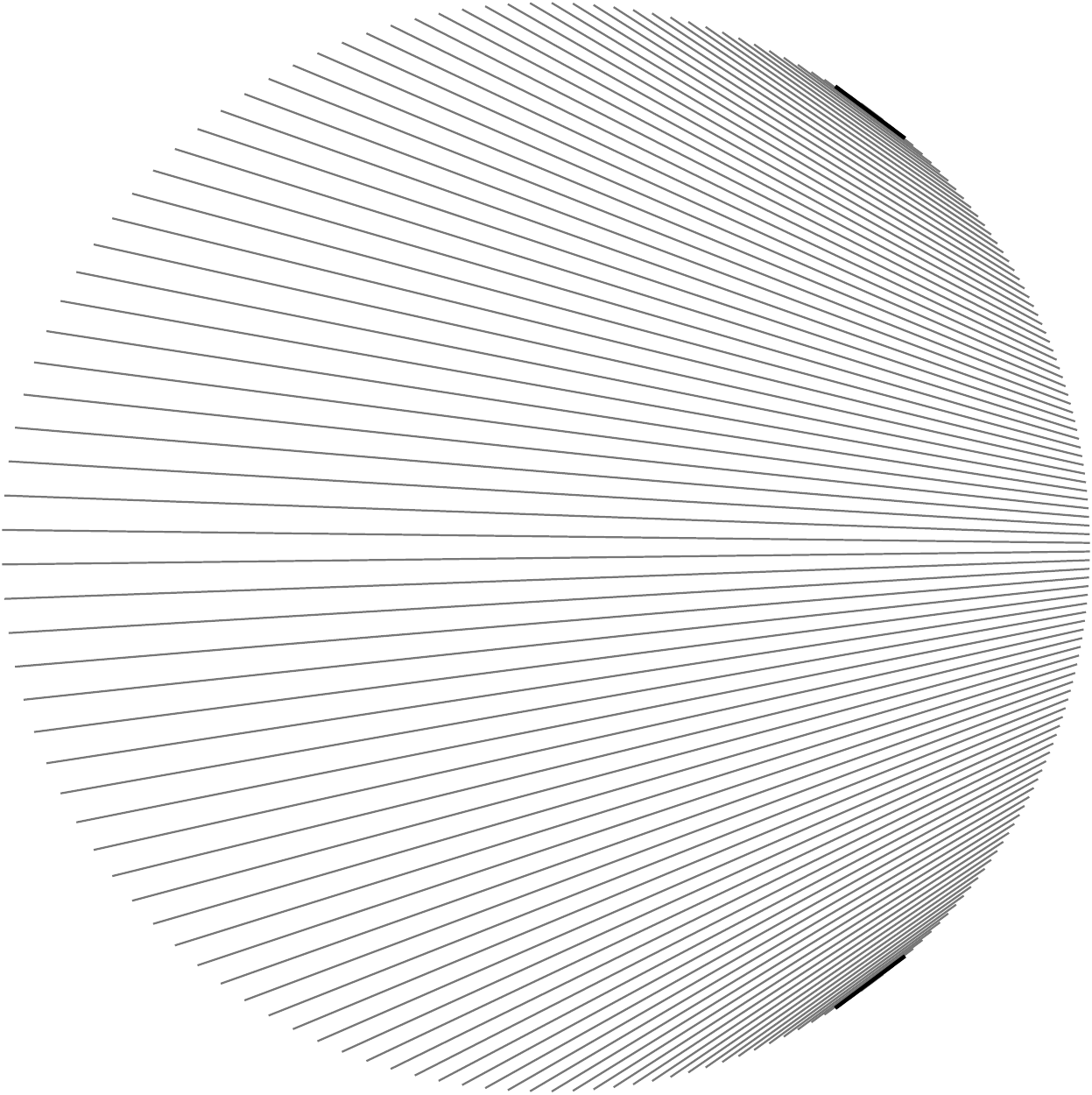}
    \caption{$YZ$ projection}
    \label{fig:yz-projection}
\end{minipage}
\end{figure}

\begin{theorem} 
By $f_{2, \sigma}: \mathbb{R}^3 \longrightarrow [0, \pi]$ let's denote the function that returns $\phi$ coordinate of a spherical coordinate system with distorted poles.\\
Let $g:[0;\pi] \longrightarrow \mathbb{R}^2$ be a function that returns the middle point of the chord corresponding to a given $\phi$ value.\\
Let $g_m:[0;\pi]\longrightarrow\mathbb{R}$ be a function that returns the slope parameter of the line spanned over the chord with the corresponding $\phi$ value. \\
Then the following equality holds:
 \begin{equation}
 g(f_{2, \sigma}([p_x, p_y, p_z])) + [1, g_m(f_{2, \sigma}([p_x, p_y, p_z]))]T = [p_y, p_z]
 \end{equation}
where $T \in \mathbb{R}.$
\end{theorem}

\begin{proof}
    $$\displaystyle a = [\cos({\frac{2\sigma}{\pi}\phi}), \sin({\frac{2\sigma}{\pi}\phi})], \ \
    \displaystyle b = [\cos({\pi + \frac{2(\pi - \sigma)}{\pi}\phi}), \;-\sin({\pi + \frac{2(\pi - \sigma)}{\pi}\phi})]$$
    $$\displaystyle g(\phi) = \frac{a + b}{2}, \ \
    \displaystyle g_m(\phi) = \frac{a_2 - b_2}{a_1 - b_1}$$
    $$\displaystyle g(\phi) = [\sin(\phi)\sin(\phi - \frac{2\sigma}{\pi}\phi), \; \sin(\phi)\cos(\phi - \frac{2\sigma}{\pi}\phi)], \ \
    \displaystyle g_m(\phi) = -\tan({\phi - \frac{2\sigma}{\pi}\phi} )   $$
    $g(\phi) + [1, g_m(\phi)]T$ is a line formula of the projected latitude at $\phi$. Therefore, by the property of the $f_2$, we get:
    $$\displaystyle g(f_{2, \sigma}([p_x, p_y, p_z])) + [1, g_m(f_{2, \sigma}([p_x, p_y, p_z]))]T = [p_y, p_z]$$
    \qed
\end{proof}

\section{Application}
The visualization of hair on the human head presents a challenging problem due to the complex geometry of the three-dimensional model on which it grows. To be exact we are interested in a two-dimensional mesh spanned over the surface of the head that contains the anatomical hair areas. Each vertex of the mesh signalizes the place of the so-called guide hairs, which preserve the information about the structure of the hairstyle.

While processing those objects it is crucial to select the correct representation of the coordinate system. We have explored the difficulties that arise when attempting to create a realistic simulation of hair and discuss here a strategy for overcoming these challenges.

The most important feature that we want from our system is maintaining a~consistent representation of the mesh, as it allows us to define closed sub-areas on the 2D map. This ensures the simplicity of image processing algorithms used in generating different assets in further parts of the simulation. Additionally, consistency guarantees the unique identification of every triangle on the mesh using its three vertices, which wouldn't be true if we allow the wrapping of the sides at the seam lines. There are many more advantages to such representation, such as easier calculation of splines, better numerical accuracy, etc.

Unfortunately, standard approaches don't meet the highlighted criteria. For the hair mesh, always at least one of the points, $x$, or its polar counterpart, $x^o$, belongs to it. Therefore, any pole positioning would contain the singular point and what follows, the seam line would always intersect with the part of the mesh. To overcome these issues, it is necessary to use the proposed method of warping the poles of the spherical coordinate system.

By distorting the poles, we can position the singular points along with the seam line outside the hair mesh's domain. This enables all the benefits we talked about earlier with the very little computational cost involved.

Figures~\ref{fig:spherical} show the process of fitting the hair mesh into one cohesive area. This is achieved by rotating the poles and then distorting them by some angle $\sigma$ to move the eccentric vertices to this compact form seen in the picture. It's important to note that this method is easily replicable to other problems involving star-shaped three-dimensional models with the intent of processing the subset of its surface.

\begin{figure}[!ht]
    \centering
    \includegraphics[width=0.8\textwidth]{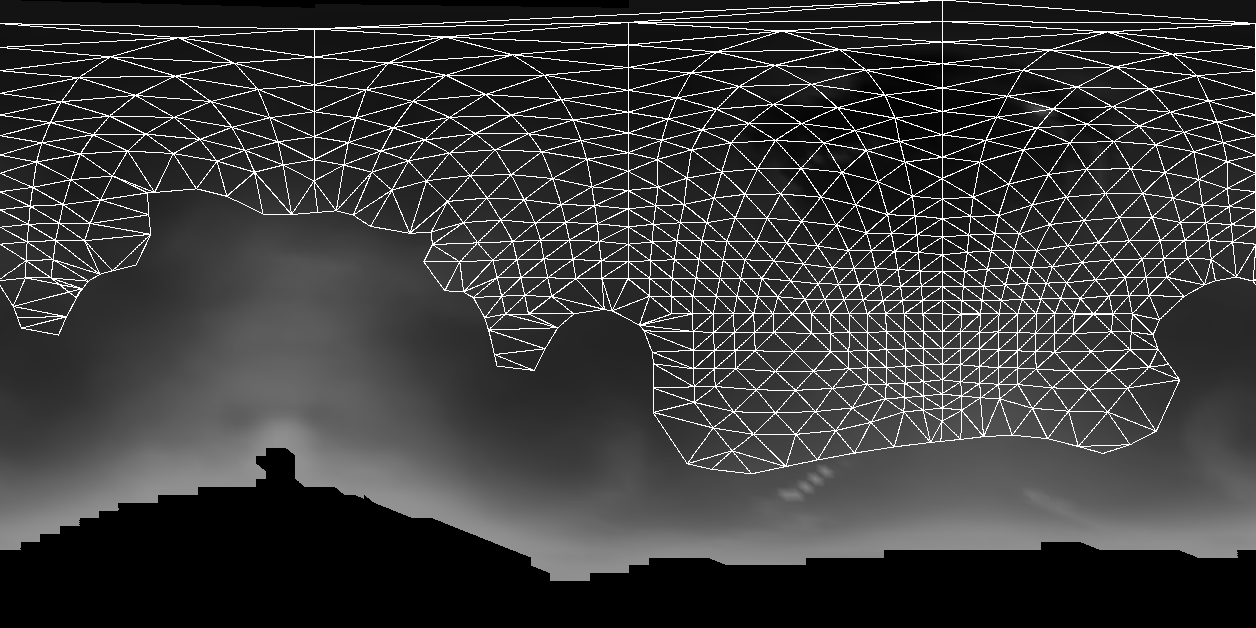}
    \includegraphics[width=0.8\textwidth]{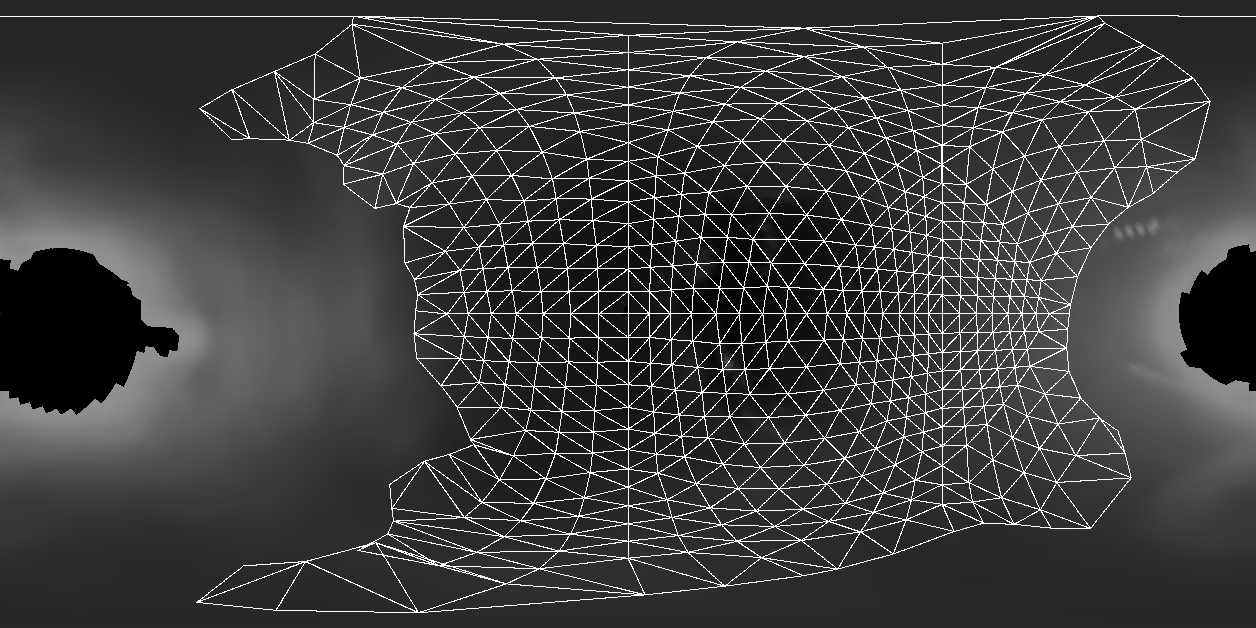}
    \includegraphics[width=0.8\textwidth]{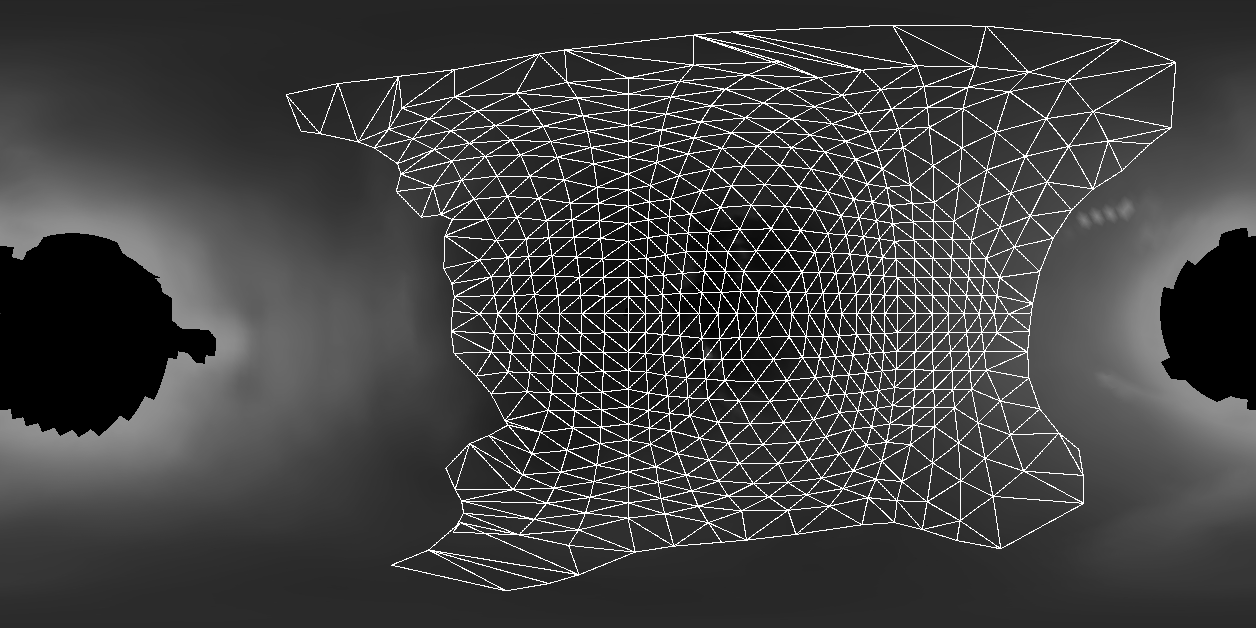}
    \caption{Consecutive transformations leading to the final mesh representation}
    \label{fig:spherical}
\end{figure}

\section{Discussion and conclusions}
The proposed methods can be extended to handle more complex shapes and mesh topologies. The methods can also be integrated into existing mesh processing frameworks, such as subdivision or smoothing algorithms, to improve their performance. The effectiveness of the proposed methods can also be compared to existing techniques to assess their benefits and limitations.

Additionally, the proposed methods can be extended to support other types of coordinate systems, such as cylindrical or polar coordinates. This would enable the representation of a wider range of shapes and improve the versatility of the mesh processing algorithms.

Furthermore, the proposed theta-phi mapping technique can be combined with other surface analysis methods, such as curvature estimation or surface parameterization, to provide a more comprehensive analysis of the mesh surface.

Another direction for future work is to investigate the impact of the proposed methods on the visual quality of the representation of the mesh. The distortion introduced to the poles of the spherical coordinate system may affect the visual appearance of the mesh, and it is important to evaluate the impact on the final result.

Finally, the proposed methods can be applied to fields beyond computer graphics, such as engineering or physics. For example, the angular measurements and the spherical coordinate system can be used to model the behavior of a~physical system that exhibits spherical symmetry.

In conclusion, the proposed methods for coordinate systems and theta-phi mapping provide a valuable contribution to the field of mesh processing and computer graphics. The methods can improve the efficiency and accuracy of various operations, while also enabling the representation of more complex shapes. Future work can extend the proposed methods to other coordinate systems, evaluate their visual impact, and apply them to other fields beyond computer graphics.

\subsubsection*{Acknowledgements.}
This work has been performed in a frame of the project entitled `Surgery planning and operational assistance system for effective hair transplant surgery', No. POIR.01.01.01-00-0597/20, financed by the funds of the National Centre for Research and Development (NCBiR) in Poland.

\bibliographystyle{splncs.bst}
\bibliography{biblio}

\end{document}